\documentclass[conference]{IEEEtran}

\usepackage{comment}
\usepackage{url}
\usepackage{algorithm}
\usepackage{algorithmic}
\usepackage{nth}
\usepackage{amssymb}
\usepackage{epsfig}
\usepackage{wrapfig}
\usepackage{enumitem}
\usepackage{subfig}
\usepackage{color}
\usepackage{multirow}

\usepackage{graphicx}

\usepackage[utf8]{inputenc}
\usepackage[english]{babel}
\usepackage{amsmath}
\usepackage{amsthm}
\newtheorem{theorem}{Theorem}
\newtheorem{lemma}[theorem]{Lemma}

\ifCLASSOPTIONcompsoc
  \usepackage[nocompress]{cite}
\else
  \usepackage{cite}
\fi
\ifCLASSINFOpdf
\else
\fi

\begin{document}




\title{FlexiComposer: Flexible Boolean Operation based Multiplex Community Composition}
\title{Efficient Community Detection in Boolean Composed Multiplex Networks}

\author{\IEEEauthorblockN{Abhishek Santra\IEEEauthorrefmark{1}, Sanjukta Bhowmick\IEEEauthorrefmark{3} and
Sharma Chakravarthy\IEEEauthorrefmark{4}}
\IEEEauthorblockA{\IEEEauthorrefmark{1}\IEEEauthorrefmark{4}IT Lab and CSE Department, University of Texas at Arlington, Arlington, Texas \\
\IEEEauthorrefmark{3}CSE Department, University of North Texas, Denton, Texas \\
Email: \IEEEauthorrefmark{1}abhishek.santra@mavs.uta.edu,
\IEEEauthorrefmark{3}sanjukta.bhowmick@unt.edu,
\IEEEauthorrefmark{4}sharma@cse.uta.edu}}

\IEEEtitleabstractindextext{%
\begin{abstract}
Networks (or graphs) are used to model the dyadic relations between entities in a complex system. In cases where there exists multiple relations between the entities, the complex system can be represented as a multilayer network, where the network in each layer represents one particular relation (or feature). The analysis of multilayer networks involves combining edges from specific layers and then computing a network property.

Different subsets of the layers can be combined. For any Boolean combination operation (e.g. AND, OR), the number of possible subsets is exponential to the number of layers. Thus recomputing for each subset from scratch is an expensive process. In this paper, we propose to efficiently analyze multilayer networks using a method that we term {\em network decomposition}.

Network decomposition is based on analyzing each network layer individually and then aggregating the analysis results. We demonstrate the effectiveness of using network decomposition for detecting communities on different combinations of network layers. Our results on multilayer networks obtained from real-world and synthetic datasets show that our proposed network decomposition method requires significantly lower computation time while producing results of high accuracy.

\end{abstract}

\begin{IEEEkeywords}
Multilayer network analysis, Community Detection, Time-Accuracy Trade Off.
\end{IEEEkeywords}}

\maketitle

\IEEEdisplaynontitleabstractindextext

%
\IEEEpeerreviewmaketitle


%
%
%
%

\section{Introduction}
\label{sec:introduction}

\IEEEPARstart{N}{etworks} (or graphs) are used to represent the pair-wise relationship between entities in a system. The entities may be connected by multiple relations (or features). For example, traffic accidents can be related if they occur in the same location, or under the same light condition; actors can be related if they acted in movies of the same genre; authors can be related if they publish in the same conferences. The relationships pertaining to each feature can be represented as a network. This group of networks are together termed as a \textit{multiplex} or {\em homogeneous multilayer network}, where each  network (layer) denotes a distinct relationship based on a feature among the same set of entities.


{\em Motivation.} Often different sets of layers need to be combined and analyzed. For example, if a city with limited resources was trying to reduce accidents, they might be interested in finding the {\em top two} conditions that lead to accidents. In this case, the combination of all pairs of layers have to be tested. Another example would be to identify which subset of social networking platforms form the most strongly connected groups. In this case, all possible combinations of the layers representing the platforms have to analyzed.  

In the current approach, multiplexes are analyzed by combining the required layers into a single network, which we term the {\em composed network}. Network analysis operations such as community detection or finding high centrality vertices is then applied to this composed network. It is easy to see that the analysis has to be recomputed from scratch for different combinations of layers, rendering the analysis process very expensive. As a solution, we propose an elegant method, which we term {\em network decomposition}, for finding communities in different compositions of the network layers. 

{\em Our Contribution.} Our main contribution is to apply network decomposition for efficiently finding communities in different Boolean compositions of network layers. The principle is to first analyze the property, such as communities, in each individual layer and then aggregate the results to obtain the final results on the composed network (see Figure \ref{fig:decoupling-HoMLN}). Thus, we only need to analyze each network once, and then combine the results as per the aggregation method. 

\begin{figure}[h]
    \centering
    \includegraphics[width=0.48\textwidth]{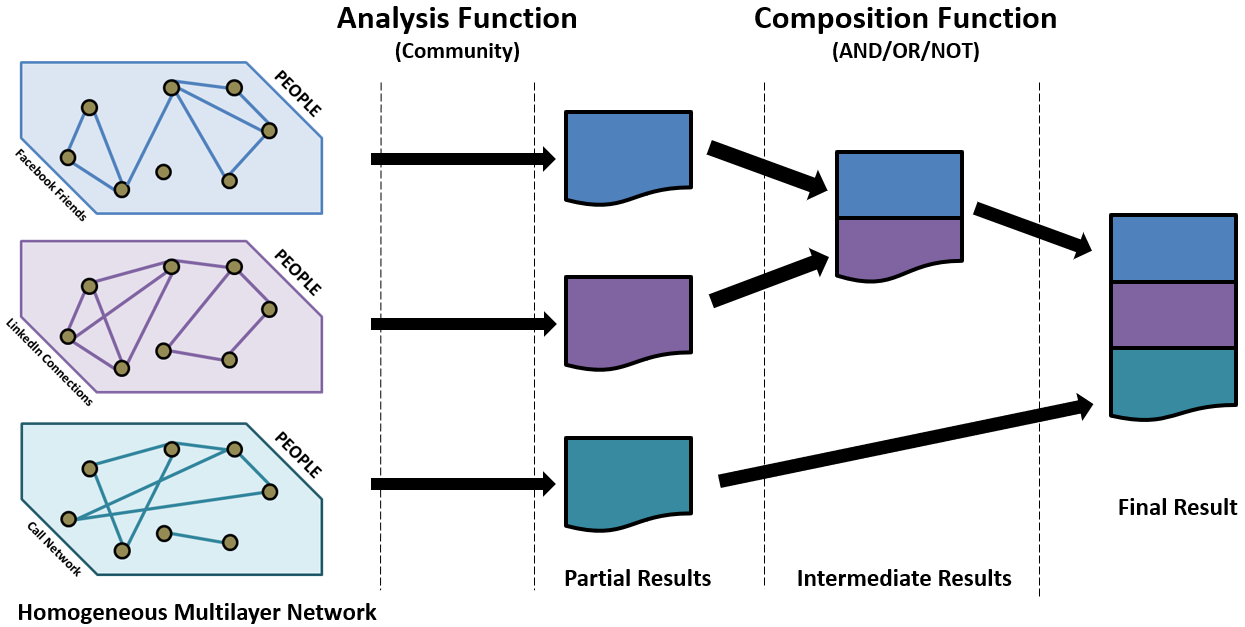}
    \caption{Illustration of network decomposition. Each  layer is analyzed separately and then the results are aggregated.}
    \label{fig:decoupling-HoMLN}
\end{figure}



Our goal is to apply network decomposition to efficiently find communities (i.e. strongly connected groups on nodes) in layers of multiplex that are combined using the Boolean AND and OR operations. For the AND (OR) operation, the combined  network will contain an edge, if there exists an edge in \textit{all} (\textit{any one}) of the individual layers. An AND-Composition represents how multiple features together affect an analysis. For example, in identifying  regions that become accident prone due to poor lighting conditions \textit{as well as} bad roads. An OR-Composition represents how any one of the features affects a property. For example, in finding the group of people who are friends via least one of the social networking platforms among Facebook, Linkedin and Twitter.




Using network decomposition, we generate communities for each of the individual layers and then combine the resultant communities based on the Boolean operators. As our results in Section~\ref{sec:experiments} demonstrate, our approach significantly reduces the computational costs  while providing results of high accuracy. 


{\em Problem Statement:} Given a set of layers $G_1, G_2, \ldots,  G_x$, that are combined using a Boolean operation $\bigoplus$ and a community detection algorithm $COMM$, develop an aggregation algorithm $\Theta$, such that $COMM(\bigoplus_{i=1}^{x}(G_i)) \approx \Theta_{i=1}^{x}(COMM(G_i)$).

Our goal is to find an aggregation algorithm $\Theta$, such that the results of finding the communities in the individual layers and then aggregating them via $\Theta$, should be the nearly the same as the  communities obtained from the composed network where the layers are combined using the Boolean operator $\bigoplus$~\footnote{We use "nearly equal" rather than "equal" because community detection is non-unique and changes to the order in which vertices are processed can slightly change the results.}.




The remainder of the paper is organized as follows. In Section~\ref{sec:multiplexIntroduction}, we provide a brief description of community detection in multiplexes. In Section~\ref{sec:contribution}, we present our contribution of community detection using network decomposition. In Section~\ref{sec:experiments} we present the experimental results. In Section~\ref{sec:relatedWork}, we discuss related research and conclude in Section~\ref{sec:conclusions} with a discussion of our future plans.
\section{Community Detection in Multiplexes}
\label{sec:multiplexIntroduction}
We describe how multiplex networks are created from datasets, and how communities are obtained in AND and OR composition of layers. 

\subsection{Creating Multiplex Networks}
In a multirelational dataset, entities are related through multiple features. For example, people can be connected through different social networks. Authors of academic papers can be connected based on the common journals where they publish.  

In a multiplex or multilayer network, each  feature is represented as a separate network or layer. The set of entities (or nodes)  remain the same in each layer. The edges, based on interactions between the entities change across the layers with respect to the corresponding feature. Table \ref{notations} lists the notations of  the terms of multilayer networks discussed in this paper.

\begin{table}[!t]
    \renewcommand{\arraystretch}{1.3}
    \caption{List of notations used for defining the concepts.}
    \label{notations}
    \centering
        \begin{tabular}{|c|c|}
            \hline
                $N_{L}$ & Number of layers \\
            \hline
                        \hline
                $I$ & Set of entities \\
            \hline
                        \hline
                $f$ & Set of features/layers \\
            \hline
                        \hline
$G(V_{k}, E_{k})$ or $G_k$ & The $k^{th}$ layer \\
            \hline
                        \hline
                $u_{k}^{i}$ & Represntative node for $i^{th}$ entity in the $k^{th}$ layer \\
			\hline
                        \hline
	$V_{k}$ & Set of nodes in the $k^{th}$ layer \\
            \hline
                        \hline
                $(u_{k}^{i}, u_{k}^{j})$ & An edge in the $k^{th}$ layer \\
                            \hline
            \hline
                $E_{k}$ & Set of edges in the $k^{th}$ layer \\
            \hline
                        \hline
$C(V^m_k, E^m_k)$or $C^m_k$ & The $m^{th}$ community in the $k^{th}$ layer \\
            \hline
             \hline
	$V^m_{k}$ & Set of nodes in $C^m_k$ \\
            \hline            \hline
	$E^m_{k}$ & Set of edges in $C^m_k$ \\
            \hline

        \end{tabular}
\end{table}

\subsubsection{Modeling the IMDb Dataset as a Multiplex}  We use the Internet Movie Database (IMDb) to illustrate how a multiplex is constructed.
The IMDb is an online database  that contains information on television programs and movies including actors, directors, genre, and year of release \cite{data/type2/IMDb}. 

We create a multiplex where the entities represent actors and two actors are connected to each other if they have acted in the same movie. Each layer in the multiplex represents a movie genre, such as comedy, drama, action, etc. 
As per the notations in Table \ref{notations}, the IMDb multiplex is defined follows:
\begin{itemize}
\item $I$: The set of actors form the entities,  $I$ = $\{I_1, I_2, ...\}$.
\item $f$: The set of features corresponds to the genres, $f$ = $\{f^1, f^2, ...\}$. The number of layers, $N_L$, is equal to $|f|$.
\item $G(V_k, E_k)$ or $G_k$: The $k^{th}$ layer represents the relations among the set of actors, $I$, with respect to the genre $f^{k}$.
\item $u_k^i$ $\in$ $V_k$: In the $k^{th}$ layer, $G_k$, the $i^{th}$ actor, $I_i$, is represented by a vertex, $u_k^i$. 
\item $(u_k^i, u_k^j)$ $\in$ $E_k$: If the $i^{th}$ and the $j^{th}$ actors have worked together in at least one movie that belongs to the $k^{th}$ genre, $f^k$, $G_k$, contains the edge $(u_k^i, u_k^j)$.
\end{itemize}

\begin{figure}[h]
    \centering
    \includegraphics[width=0.48\textwidth]{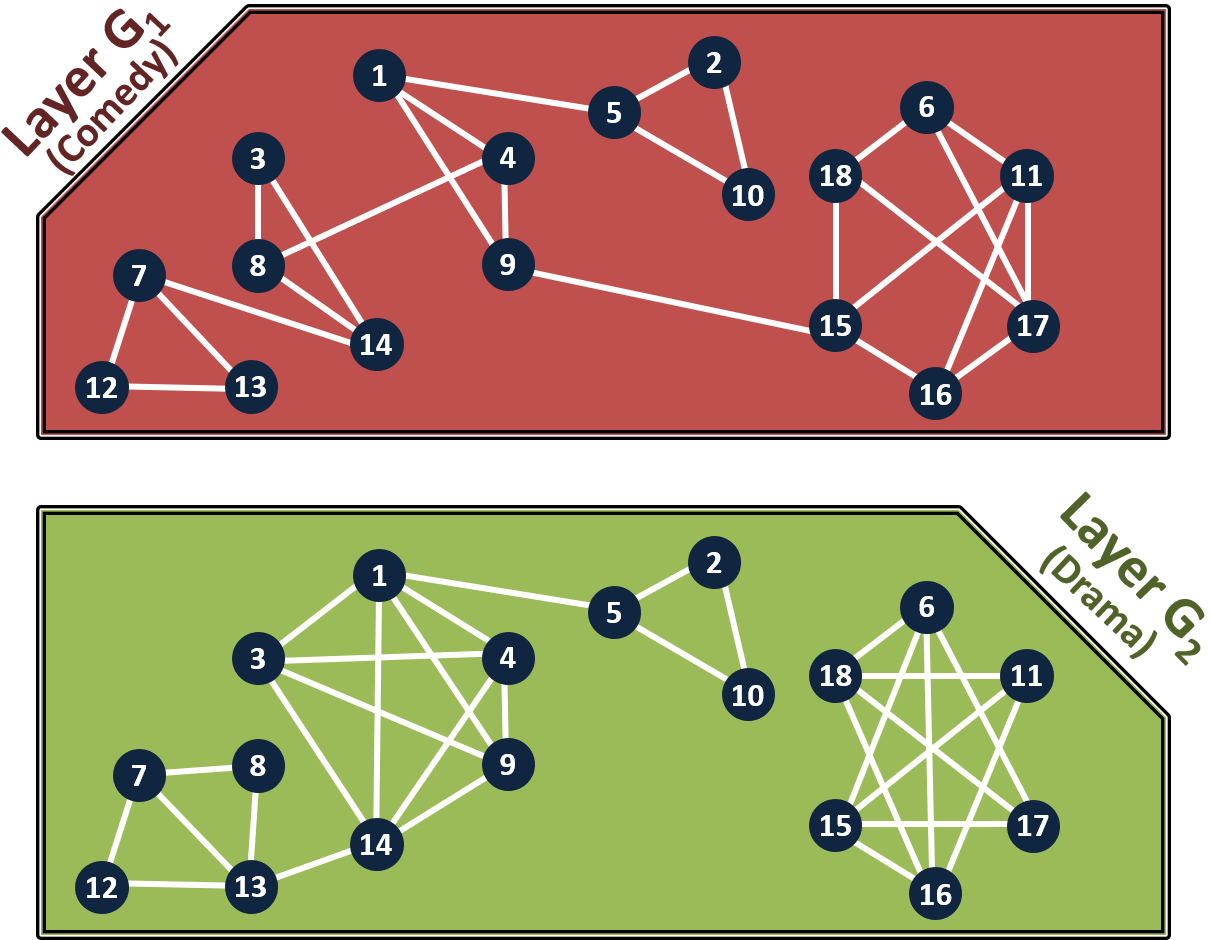}
    \caption{Example of the IMDb Multiplex for co-actors with 16 actors and two genres, comedy and drama.}
    \label{fig:genreMultiplex}
\end{figure}

In Figure \ref{fig:genreMultiplex} we have selected two genres, comedy ($f^1$) and drama ($f^2$) to form the two layers, $G_1$ and $G_2$, respectively. This multiplex shows the co-actor relationship among 16 actors (denoted by nodes numbered from 1 to 16) with respect to these  genres. The same 16 actors are present in both layers. Note that each co-actor network has a distinct structure. By taking the information from the two networks together we can gain interesting insights to the data, as follows.

For example, actors $I_3$ and $I_8$ have never worked together in a drama, but have worked together in a comedy. Thus this pair of actors may together be more likely to be considered in a comedy, rather than a drama.  Also observe
that the actor $I_{14}$ is the actor with most connections in the drama genre, while in the case of comedies, actor $I_{11}$ is one of the nodes with the most connections, i.e. worked with most number of actors. 



\subsection{Community Detection in Multiplexes} Community detection involves finding items with similar properties by identifying tightly connected groups of vertices. We consider non-overlapping communities, that is, there are no common vertices or edges between  two communities. Figure \ref{fig:IMDbComposition} shows the communities for the composed layers, $G_{1AND2}$ and $G_{1OR2}$ for IMDB multiplex in Figure \ref{fig:genreMultiplex}. 

{\em Bridge Edges.} We term the external edges that connect two communities as {\em bridge edges}. Formally, if there exists an edge, $(u_k^i, u_k^j)$, such that $u_k^i$ $\in$ $C_k^m$ and $u_k^j$ $\in$ $C_k^n$, where $m \neq n$, then  this edge is a bridge edge. Bridge edges form links between two distinct communities. In the AND and OR composed-layer of Figure \ref{fig:IMDbComposition}, the actors $I_1$ and $I_5$ belong to different communities, but are connected by a bridge edge. The actors $I_9$ and $I_15$ have a bridge edge in the OR composed network, but not in the AND-composed network.

\subsubsection{Communities in AND-Composed Layers.} AND composition of layers in a multiplex allows users to find communities that are related across multiple features. Examples of some queries that can be addressed by the AND composition are;

\begin{itemize}
\item Groups of actors who have expertise in working together in {\em both} comedies \textbf{and} dramas (IMDb multiplex).
\item Author groups who publish in all of {\em these} conferences; ICDM, SIGMOD \textbf{and} VLDB (DBLP multiplex).
\item Groups of accidents that have similar conditions for {\em all  these} features;  light conditions, weather conditions, road conditions, and speed limit (Accident Multiplex).
\end{itemize}
The standard practice is to combine the layers using the AND operation, i.e. only edges that occur in all the layers are included. Then a community detection algorithm, such as Infomap, is executed on the combined network.  This single graph approach, termed C-SG-AND, is given in Algorithm~\ref{naive:AND}.


\begin{algorithm}[H]

        \begin{algorithmic}[1]
				\label{algo:NAIVE-AND}
				\REQUIRE Layers $G_1, G_2, \ldots G_x$
				\ENSURE return $L_{1,2,\ldots,x}^{AND}$ - a list of communities
				\STATE \quad $G_{1 AND 2 \ldots AND x}$ $\leftarrow$ $\{G_1$ AND $G_2$ \ldots AND $G_x$\}
				
				\COMMENT{ $G_{1 AND 2 \ldots AND x}$ contains edges that are in all the networks $G_1$, $G_2$, \ldots, $G_j$.}
				\STATE \quad $L_{1,2,\ldots,x}^{AND}$ = COMM($G_{1 AND 2 \ldots AND x}$)
				
				\COMMENT{Find communities in $G_{1 AND 2 \ldots AND x}$.}
			\end{algorithmic}
        \caption{Algorithm for C-SG-AND}
        \label{naive:AND}
		\end{algorithm}
\subsubsection {Communities in OR-Composed Layers} OR-composition forms a composed network that includes an edge if it appears in any of the layers. Algorithm~\ref{naive:OR} shows the steps of this single network based community detection using the OR operation, termed as C-SG-OR. Examples of queries that can be addressed by the OR composition are;
\begin{itemize}
\item Groups of actors who have acted together in either a comedy \textbf{or} drama (IMDb multiplex).
\item  Groups of authors who have published in \textbf{at least one} of these conferences, ICDM, VLDB, SIGMOD (DBLP multiplex).
\item Groups of accidents that have at \textbf{least one} condition in common (Accident Multiplex).
\end{itemize}

\begin{algorithm}[H]

        \begin{algorithmic}[1]
				\REQUIRE Layers $G_1, G_2, \ldots G_x$
				\ENSURE return $L_{1,2,\ldots,x}^{OR}$ - a list of communities
				\STATE \quad $G_{1 OR 2 \ldots OR x}$ $\leftarrow$ $\{G_1$ OR $G_2$ \ldots OR $G_x$\}
				
				\COMMENT{ $G_{1 OR 2 \ldots OR x}$ contains edges that are in \textit{at least one} of the networks $G_1$, $G_2$, \ldots, $G_x$.}
				\STATE \quad $L_{1,2,\ldots,x}^{OR}$ = COMM($G_{1 OR 2 \ldots OR x}$)
				
				\COMMENT{Find communities in $G_{1 OR 2 \ldots OR x}$.}
			\end{algorithmic}
        \caption{Algorithm for C-SG-OR}
        \label{naive:OR}
		\end{algorithm}
		
\begin{figure}[h]
    \centering
  \includegraphics[width=0.48\textwidth]{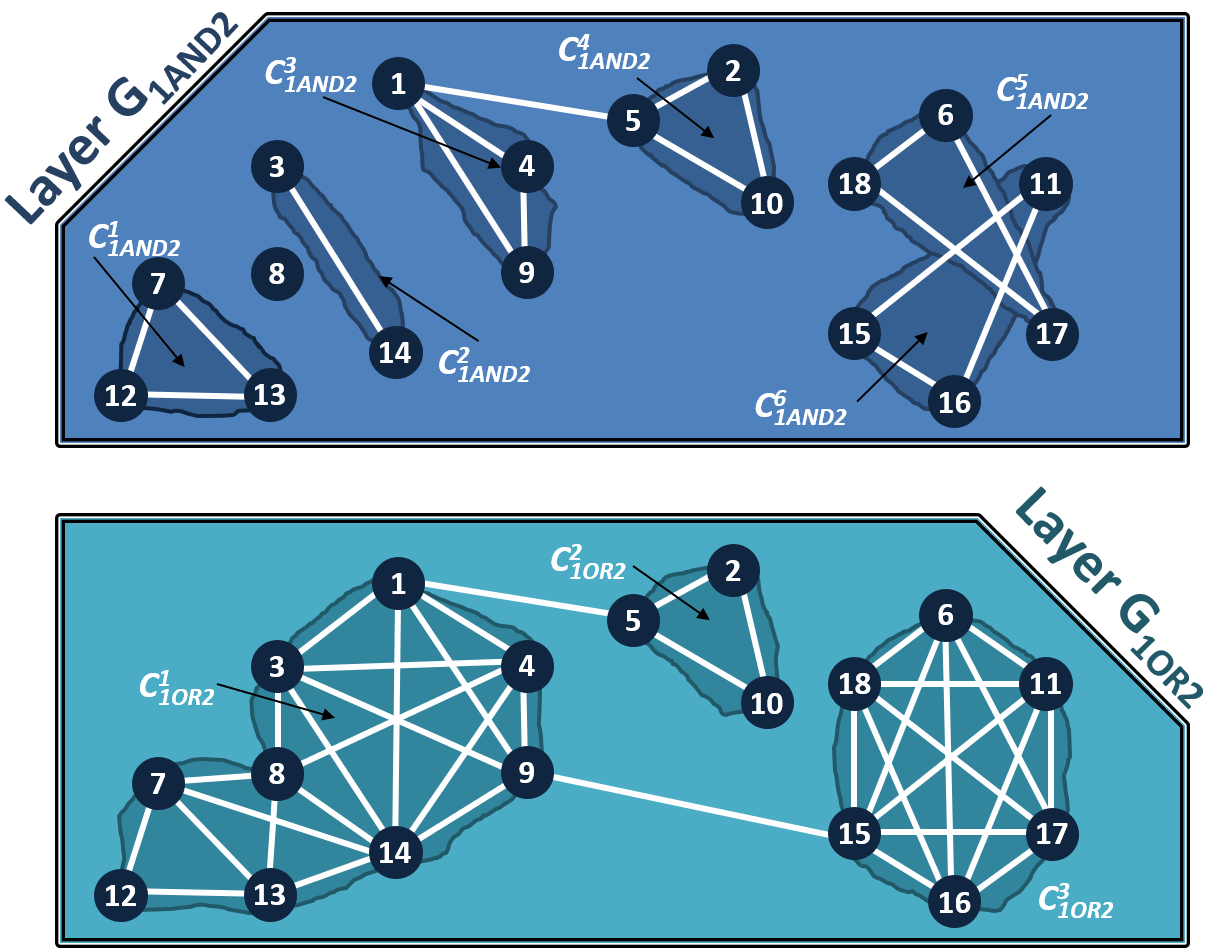}
    \caption{Composed Layer Communities of the IMDb Multiplex shown in Figure \ref{fig:genreMultiplex}}
    \label{fig:IMDbComposition}
\end{figure}

\begin{figure}[h]
    \centering
  \includegraphics[width=0.48\textwidth]{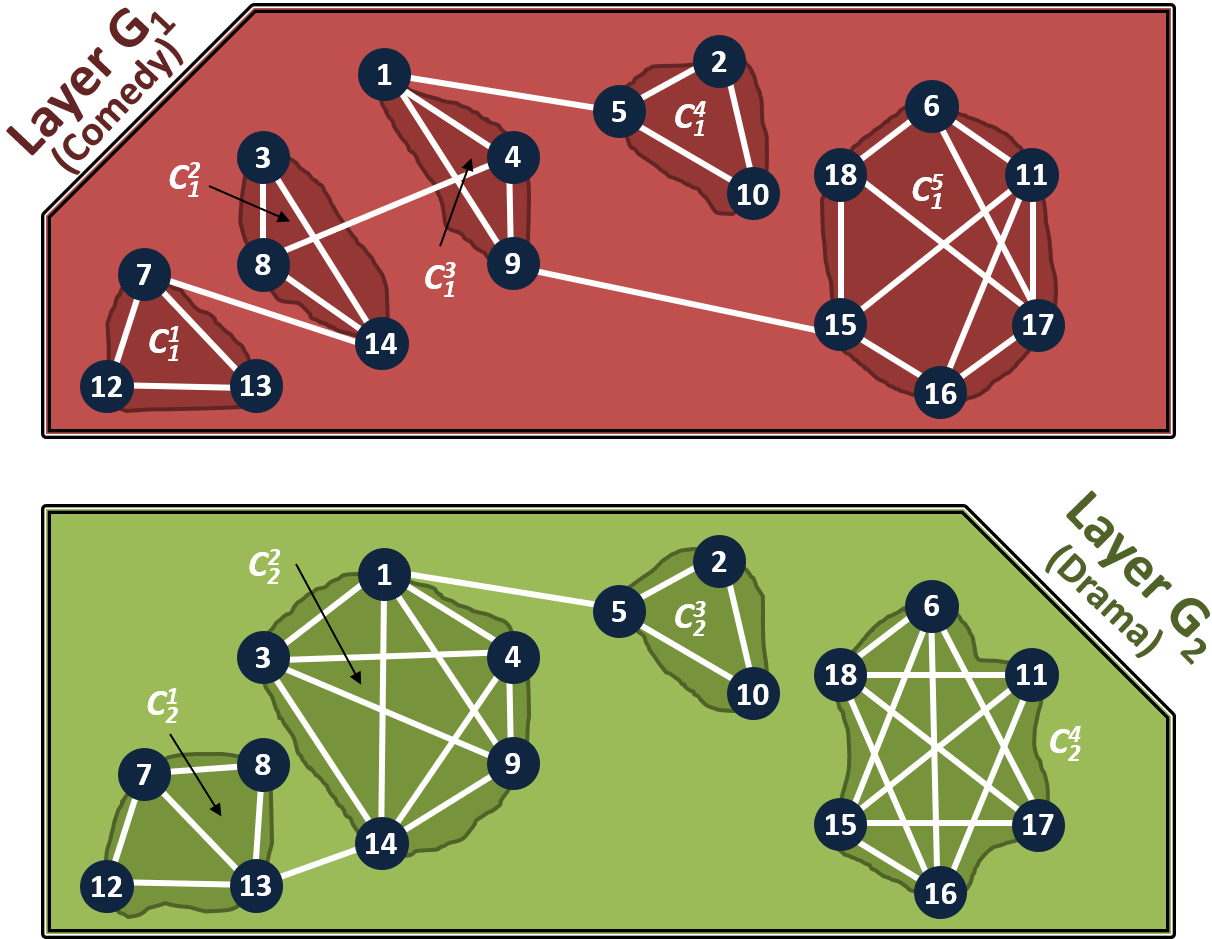}
   \caption{Communities in each layer of the IMDb Multiplex}
    \label{fig:IMDbCommunities}
\end{figure}

\begin{algorithm} [tbh]
        \begin{algorithmic}[1]
				\label{algo:CV-AND}
				\REQUIRE Communities from layers $G_i$ and $G_j$: \\ COMM($G_i$) = \{$C_i^1(V_i^1, E_i^1)$, $C_i^2(V_i^2, E_i^2)$, \ldots, $C_i^x(V_i^x, E_i^x)$\}, \\
				COMM($G_j$) = \{$C_j^1(V_j^1, E_j^1)$, $C_j^2(V_j^2, E_j^2)$, ..., $C_j^y(V_j^y, E_j^y)$\}
				\ENSURE return $L_{i,j}^{CV-AND}$ - a list of communities
				
				\STATE \quad $L_{i,j}^{CV-AND}=\Phi$ 
				
				\COMMENT {Initialize the set of communities to NULL.}
				\FOR {each community pair say, $C_i^p$ and $C_j^q$}
\STATE \quad $C_{i,j}^{p,q}$ =($V_i^p$ $\cap$ $V_j^q$)

\COMMENT {Create new combined community by taking the common {\bf vertices} of every pair of communities.}
				\STATE \quad $L_{i,j}^{CV-AND}$ = $L_{i,j}^{CV-AND}$ $\cup$ $C_{i,j}^{p,q}$
				
				\COMMENT {Add new community to the set of communities.}
				
				\ENDFOR
			\end{algorithmic}
        \caption{Algorithm for CV-AND}
        \label{CV_AND}
		\end{algorithm}

\section{Our Contribution: Network Decomposition for Efficient Community Detection on Multiplex}
\label{sec:contribution}
The Boolean composition of the layers of a multiplex provides in-depth analysis of the database. However, for any single Boolean operation, say AND, $2^N-1$ different combinations are possible. Thus the cost of  finding the communities on each of them separately is very expensive. Moreover, if the networks do not change, several computations are rendered redundant. For example, consider finding the communities in the composed layer $G_{1AND2AND3}$ and $G_{1AND2AND4}$.  In this case, the composed layer related to $G_{1AND2}$ remains unchanged, but has to be recomputed. 

As a solution, we propose network decomposition for efficient community detection on multiplex networks. In network decomposition, the communities in each layer are identified separately and the results are then aggregated to obtain the results with respect to the composed network. Note that the storage required is only of the order of $O(V*f)$, where $V$ is the number of vertices in each layer and $f$ is the number of features/layers. Figure~\ref{fig:IMDbCommunities} shows the  communities in each of the layers of the example IMDB network.

The {\bf challenge} is to develop aggregation algorithms, $\Theta$, that can correctly aggregate the communities from each of the layers to obtain the communities over the composed network. We present the aggregation methods for AND and OR composition. For ease of understanding we will discuss the algorithms with respect to two layers. Note, however, that any binary operations can be easily extended to multiple layers.

\subsection{Vertex based Community Detection of AND Composed Layers (CV-AND)} An earlier paper on obtaining communities in AND-composed layers was presented in~\cite{ICCS/SantraBC17}. We discuss this work, termed CV-AND, here for completeness. 

CV-AND (see Algorithm~\ref{CV_AND}) heavily depends on the presence of {\em self-preserving communities}. A community is self preserving if the vertices in it are so tightly connected such that any connected subset of the vertices will form a smaller community rather than joining an existing larger community. 
Formally, consider a graph $G_i$, with a community whose vertices are given by the set $C_v$. Now consider a subset of vertices $C^S_v \in C_v$. If the vertices in $C^S_v$ form a community by themselves, for {\em any subset} $C^S_V$ of $C_v$, where $\|C^S_v\| \ge 3$, then community $C_v$ is self preserving. The main result of ~\cite{ICCS/SantraBC17} was that if the communities from the layers are self preserving, then the communities of the AND-composed graph can be obtained by taking the vertex based intersection of the communities from the individual layers. 

{\em Drawbacks} The main drawback of CV-AND is that for most networks, there is no guarantee that the communities will be self-preserving. If this algorithm is applied without testing for self-preserving communities, the results may not be accurate. 

As an example, consider the community $C_1^5$ in the comedy layer of the network (Figure~\ref{fig:IMDbCommunities}). This community is not self preserving, and when combined with community $C_2^4$ in the drama layer, which has the same vertices, it gives one large community, \{ $I_6$, $I_{11}$, $I_{15}$,$I_{16}$,$I_{17}$,$I_{18}$\}. In reality, as seen in Figures~\ref{fig:IMDbComposition}, two separate communities are formed, \{ $I_6$,$I_{17}$,$I_{18}$\} and \{ $I_{11}$, $I_{15}$,$I_{16}$\}.

This is a subtle but important difference because the community id determines whether two entities are similar. If two disconnected groups of vertices are placed in the same community (as is possible when using CV-AND), then, two dissimilar groups are marked to be similar, which is incorrect.

\begin{algorithm}[tbh]
        \begin{algorithmic}[1]
				\REQUIRE Communities from layers $G_i$ and $G_j$: \\ COMM($G_i$) = \{$C_i^1(V_i^1, E_i^1)$, $C_i^2(V_i^2, E_i^2)$, ..., $C_i^x(V_i^x, E_i^x)$\}, \\
				COMM($G_j$) = \{$C_j^1(V_j^1, E_j^1)$, $C_j^2(V_j^2, E_j^2)$, ..., $C_j^y(V_j^y, E_j^y)$\}
				\ENSURE return $L_{i,j}^{CE-AND}$ - a list of communities

				\STATE \quad $L_{i,j}^{CE-AND}=\Phi$ 
				
				\COMMENT {Initialize the set of communities to NULL.}
				\FOR {each community pair say, $C_i^p$ and $C_j^q$}
				\STATE \quad \{$C_{i,j}^{p,q}$\} = ($E_i^p$ $\cap$ $E_j^q$)

\COMMENT{Create \textit{list} of k new communities by taking the common {\bf edges} of every pair of communities.}
				
				\STATE \quad $L_{i,j}^{CE-AND}$ = $L_{i,j}^{CE-AND}$ $ \cup \{C_{i,j}^{p,q}$\}
				
					\COMMENT{Add new communities to the set of communities.}
				\ENDFOR
			\end{algorithmic}
        \caption{Algorithm for CE-AND}
        \label{algo:CE-AND}
		\end{algorithm}

\subsection{Edge based Community Detection of AND Composed Layers (CE-AND)} We address these limitations by developing a community detection method, CE-AND (see Algorithm ~\ref{algo:CE-AND}), that is based on the intersection of {\em edges} rather than {\em vertices} as follows. 

For every pair of communities, $C_i^m (V_i^m, E_i^m)$ from layer $G_i$ and $C_j^n (V_j^n, E_j^n)$ from layer $G_j$,
the edge-based community intersection, $E_i^m \cap E_j^n$, will produce k disconnected edge-sets, $E_{iANDj}^1, E_{iANDj}^2, ..., E_{iANDj}^k$. These edge sets will form the AND-composed communities, $C_{iANDj}^1, C_{iANDj}^2, ..., C_{iANDj}^k$. 

{\em Correctness and Limitations.} Figure~\ref{fig:IMDBCE-AND} shows how the communities are obtained for the example network using CE-AND. Comparing this result to that in Figure~\ref{fig:IMDbComposition}, we see that most of the communities are obtained with the exception of the singleton node 8 and the common bride edge (1, 5).

\begin{lemma}
\label{lemma:1}
Let $V_{i,j}$ be a set of vertices that are assigned to the same community in all the layers of the multiplex.  Let the set of common edges in all the layers whose both endpoints are in $V_{i,j}$, be $E_{i,j}$. If $|E_{i,j}| \ge |V_{i,j}|-1$, then the community composed of vertices $V_{i,j}$ will be detected by CE-AND.
\end{lemma}
{\em Proof.} The proof  is based on the observation that if a subgraph has $n$ vertices and $n-1$ edges, then the subgraph is connected. Since this subgraph is a community and is present in all the layers, therefore the community will be detected.
Lemma~\ref{lemma:1} provides a lower bound ($n-1$) on the density of common edges of the communities that is needed for them to be detected by the CE-AND. 

\begin{lemma}
The CE-AND method returns all the communities of size $\ge 2$ that were returned by the CV-AND method. 
\end{lemma}

{\em Proof.}
The proof is derived from Lemma~\ref{lemma:1}. All self-preserving communities have to  be densely connected, and thus satisfying the lower bound. Thus CE-AND will detect all the common self preserving communities.
\begin{lemma}
\label{lemma:3}
Common edges that form a bridge edge in at least one layer will not be detected by CE-AND.
\end{lemma}
{\em Proof.}
This proof follows from Algorithm~\ref{algo:CE-AND} in that both endpoints of the common edge have to be in the same community to be considered.

Lemma~\ref{lemma:3} highlights the limitations of finding communities in cases where the common communities do not share a dense substructure. Note that these limitations only occur for communities that are not dense. Since communities are informally defined as tightly connected groups of vertices, these sparse subpgraphs in most cases may not be relevant communities at all. We illustrate the issues in Figure~\ref{fig:bridge-AND}.

The left-hand panels of Figure~\ref{fig:bridge-AND} shows two layers. The top right panel shows the communities obtained by the standard single network approach (C-SG-AND). The bottom right panel shows the communities obtained by  CE-AND.

\begin{figure}
    \centering
  \includegraphics[width=0.48\textwidth]{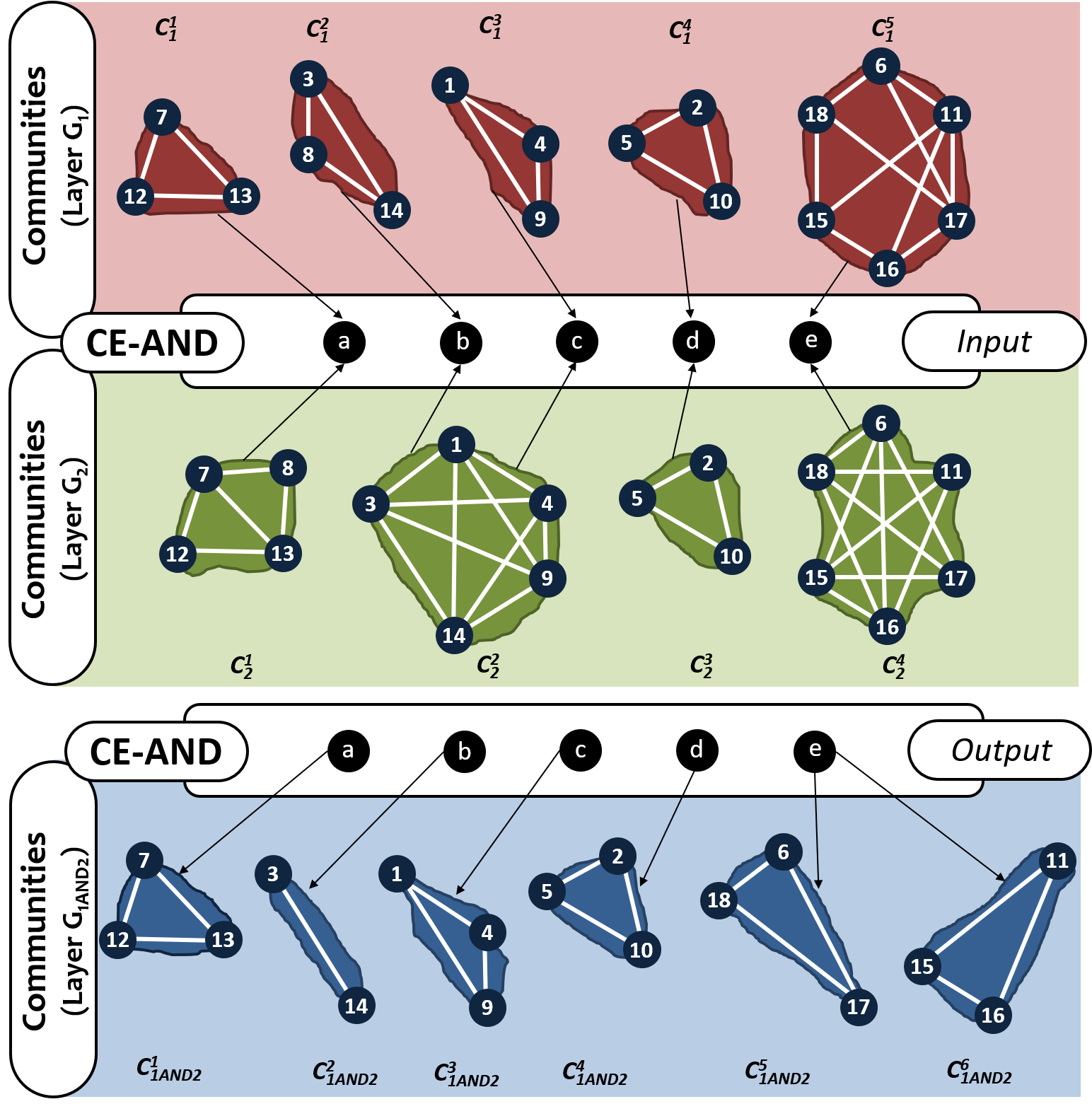}
    \caption{AND-Composition Communities of the IMDb Multiplex shown in Figure \ref{fig:genreMultiplex}, using CE-AND method}
    \label{fig:IMDBCE-AND}
\end{figure}
\begin{figure}
    \centering
  \includegraphics[width=0.45\textwidth]{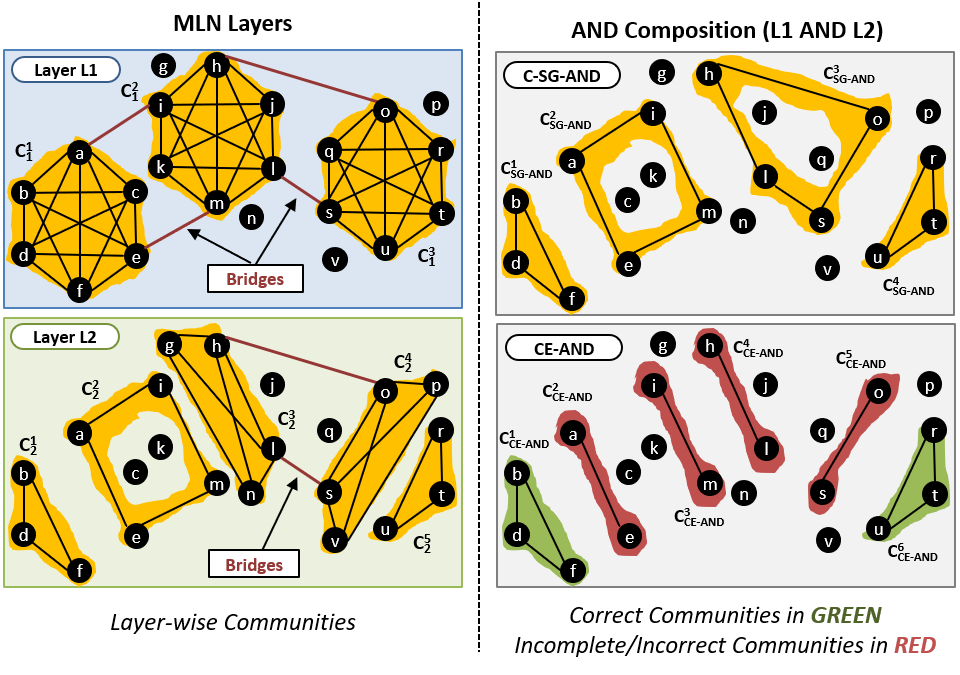}
    \caption{Effect of Bridge Edges on AND Composition}
    \label{fig:bridge-AND}
\end{figure}

Note that the community $C_{SG-AND}^3$ produced by C-SG-AND contains the edges (h, o) and (l, s) that act as bridges in both Layer L1 and L2. Thus CE-AND is not able to detect this community, and instead produces two communities, $C_{CE-AND}^4$ and $C_{CE-AND}^5$, which should be merged into one by taking the bridge edges into account.

Also consider the community $C_{SG-AND}^2$ which consists of the edges (a, i) and (e, m) that are bridges in Layer L1, but are part of the community $C_2^2$ in Layer L2. As only those edges that are within community \textit{ in all layers} are considered, CE-AND produces two communities, $C_{CE-AND}^2$ and $C_{CE-AND}^3$.

\subsection{Edge  based  Community  Detection  of  OR  Composed
Layers}

We now consider how to obtain communities in OR-composed networks. Note that the number of edges in the OR-composed network will at least as much, generally more, than the number of edges in each layer. Thus for any two layers $G_i$ and $G_j$, the total number of edges is $|E_i \cup E_j|$.

The computational complexity of community detection algorithms are at least proportional to the size of the graph. Thus the denser the graph, the more time will be required to find the communities. Thus for the OR-composed case, the goal is not only to lower the time by reducing the need to recompute different compositions of layers, but also to reduce the size of the graph to be analyzed.

To recreate the communities of OR Composed Layers, we propose the CE-OR algorithm (given in Algorithm~\ref{algo:CE-OR} and illustrated in Figure~\ref{fig:OR-example}). The CE-OR method reduces the size of the graph to be analyzed by leveraging the fact that the common communities across layers can be processed as a single node. The steps of the CE-OR algorithm are as follows;

\begin{algorithm}
        \begin{algorithmic}[1]
				\REQUIRE Communities from layers $G_i (V, E_i)$ and $G_j(V, E_j)$: \\ COMM($G_i$) = \{$C_i^1(V_i^1, E_i^1)$, $C_i^2(V_i^2, E_i^2)$, ..., $C_i^x(V_i^x, E_i^x)$\}, \\
				COMM($G_j$) = \{$C_j^1(V_j^1, E_j^1)$, $C_j^2(V_j^2, E_j^2)$, ..., $C_j^y(V_j^y, E_j^y)$\}
				\ENSURE return $L_{i,j}^{CE-OR}$ - a list of communities

\COMMENT{ Find common communities using CE-AND}\\
				\STATE Apply CE-AND on $COMM(G_i)$ and $COMM(G_j)$ to get $L_{i,j}^{CE-AND}$ 
				\\
				\textbf{Construct} \textit{OR-MG}$(V_{OR-MG}, E_{OR-MG})$
				
				\COMMENT{Assign nodes of each common community as a meta node}\\
				\FOR {each community $C_k (U_k, E_k)$ $\in$ $L_{i,j}^{CE-AND}$}
				\STATE $V_{OR-MG}$ = $V_{OR-MG}$ $\cup$ $U_k$
				\ENDFOR
				
					\COMMENT{ Assign the vertices not in any common community as a meta node}\\
				\FOR {each vertex $u \notin C_k$ $, \forall C_k \in 
				L_{i,j}^{CE-AND}$}
				\STATE $U_k=\phi$ \COMMENT {Create null set}
				\STATE {$U_k=U_k \cup u$} \COMMENT {Add $u$ to the set}
				\STATE $V_{OR-MG}$ = $V_{OR-MG}$ $\cup$ $U_k$
				\ENDFOR
				
				\COMMENT {Add Edges in the metagraph. Two metanodes, $(U,V)$ are connected if there is an intra-community edge from one constituent node of $U$ to a constituent node of $V$ in any one of the layers.}
				\FORALL { all metanode pairs $(U, V )\in V_{OR-MG}$}
                        \IF {$\exists$ $u,v,r$: $(u,v)$ $\in$ $E_i^r$ or $(u,v)$ $\in$ $E_j^r$,  $u$ $\in$ $U$ and $v$ $\in$ $V$}
                    \STATE $E_{OR-MG}$ = $E_{OR-MG} \cup (U,V)$
                        \ENDIF
				\ENDFOR				
                \STATE Insert weights on the edges of OR-MG 
                \STATE L = COMM(OR-MG)
                \STATE Expand the \textit{community representative nodes} in each community from L to get $L_{i,j}^{CE-OR}$
			\end{algorithmic}
        \caption{Algorithm for CE-OR}
		\label{algo:CE-OR}
		\end{algorithm}

{\em Overview of CE-OR.} Find the common communities in all the network layers (Line 1) by using  CE-AND.
Then construct a metagraph (OR-MG), as follows. Each metanode represents {\em a set} of vertices. Combine each of the vertices in a common community into a metanode (Line 2-4). Assign all remainder vertices, that are not part of any common community into singleton sets. Each of these sets is also a metanode (Line 5-9). Connect two metanodes, $(U,V)$ via a metaedge, if there exists an intra-community (within community)  edge, in {\em any one} of the layers between an element (node) of $U$ and an element (node) of $V$ (Line 10-14). Apply appropriate weights to these edges (Line 15). Apply community detection on the metagraph (Line 16). The communities in the OR-composed network are obtained by expanding the metanodes in the communities obtained by the CE-OR algorithm.

{\em Assigning Weights to Metaedges.}  Note that the metanodes represent vertex sets of varying sizes, and the number of edges between them represent the degree of similarity. Therefore although the original graph is composed on unweighted edges, the edges in the metagraph have to be weighted to quantify the extent of this similarity. A critical component of the CE-OR algorithm is based on correctly assigning  these weights. 

We propose two different weight metrics to quantity the similarity between two meta nodes. For any meta edge $(A,B)$, let $V_A$ and $V_B$ be the set of entity nodes in the AND-composed communities, respectively. Further, let the set of intra-community edges between $V_A$ and $V_B$ be $E_{A,B}$. Then the weight to the metaedge can be computed as follows;
\begin{itemize}
    \item {\em Aggregation:} The weight $w_a$ is the number of edges between the two communities; $w_a(A, B) = {|E_{A,B}|}$
    \item  {\em Fractional:} The weight $w_f$ is the fraction of connected nodes between the  two communities; $w_f(A, B) = \frac{|E_{A,B}|}{|V_A|*|V_B|}$. 
\end{itemize}
{\em Correctness and Limitations:} Figure \ref{fig:OR-example} illustrates how the CE-OR algorithm is applied to identify communities in the OR-composed layers of the example IMDb graph. First the CE-AND communities obtained in Figure \ref{fig:IMDBCE-AND} and the remaining vertex $I_8$ are used to form the metanodes (Figure \ref{fig:OR-example} (a)). Then these nodes are connected based on the intra-community edges. These edges are weighted in  the meta graph using $w_f$ (Figure \ref{fig:OR-example} (b)). A community detection algorithm on the metagraph produces the communities of the OR-composed layers (Figure \ref{fig:OR-example} (c)). Comparing with the communities obtained by the C-SG-OR method in Figure~\ref{fig:IMDbComposition}, to those obtained by expanding the communities in the metanodes (Figure \ref{fig:OR-example} (d)), we see that all the communities have been obtained. However, the bridge edges between the communities are missing.

\begin{figure}[h]
    \centering
  \includegraphics[width=\linewidth]{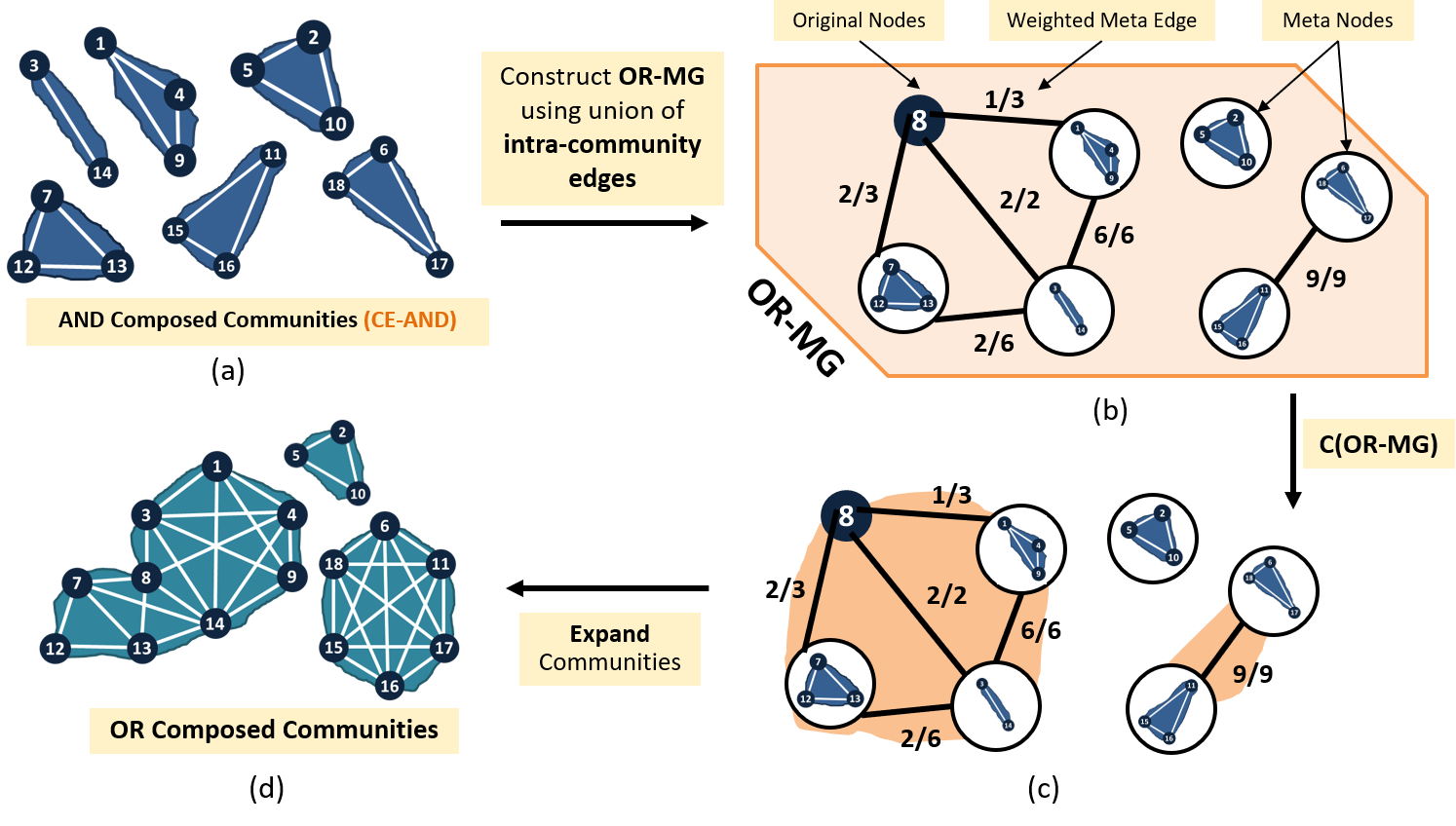}
    \caption{Illustration of CE-OR Algorithm on the IMDB Example Graph}
    \label{fig:OR-example}
\end{figure}

\begin{lemma}
Let $V_k$, $|V_k|=n$, be a set of vertices. Let $E_k$ be the set of edges whose both endpoints are in $V_k$, i.e. $\forall (x,y) \in E_k$, $x \in V_k$ and $y \in V_k$. If all $(x,y) \in E_k$, form an intra-community  edge in {\bf at least} one of the network layers and  the vertices in $V_k$ form a community in the OR-composed network, then this community will be detected by the CE-OR algorithm.
\end{lemma}

\begin{proof}
As per Algorithm~\ref{algo:CE-OR}, any edge that is an intra-community edge in at least one of the layers will be included in the metagraph. If a community in the OR-composed layer is formed only of intra-edges, then in the metagraph, given appropriate weighting function, these set of edges would also form a dense subgraph and hence a community.
\end{proof}

\begin{lemma}
If a community in the OR-composed layer contains significant number of edges that are bridge (inter-community) edges in {\bf all} the layers, then the community cannot be detected in its entirety by the CE-OR algorithm.
\end{lemma}

\begin{proof}
Since the CE-OR algorithm analyzes only the intra-community edges, the bridge edges will not be part of the metagraph. Thus any community that is composed of mainly bridge edges will not be part of the metagraph and hence will not be detected by CE-OR.
\end{proof}

\begin{figure}[h]
    \centering
  \includegraphics[width=\linewidth]{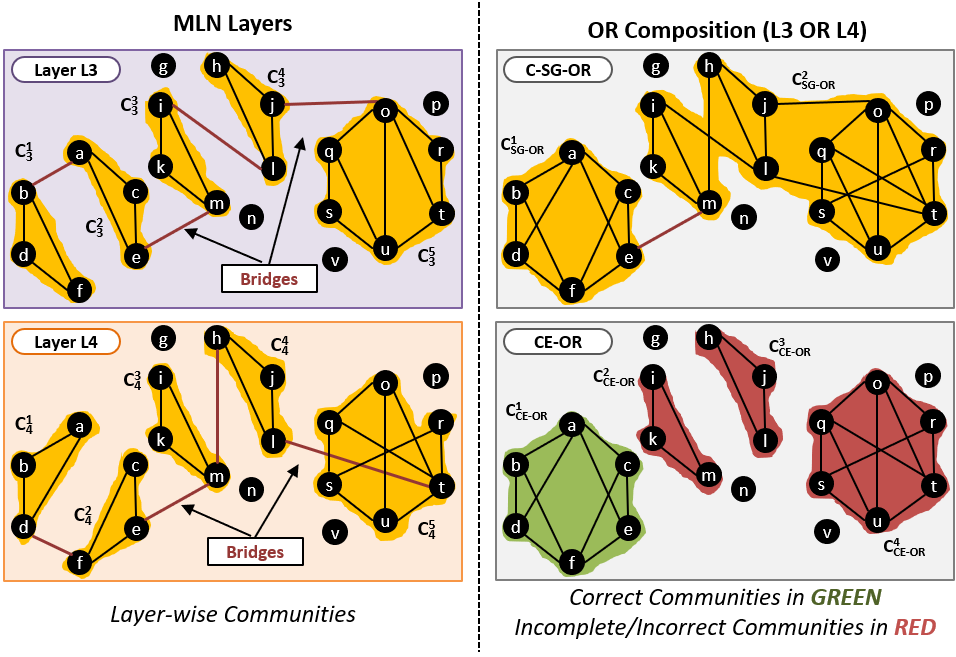}
    \caption{Effect of bridge edges on OR Composition}
    \label{fig:bridge-CE-OR}
\end{figure}

Figure \ref{fig:bridge-CE-OR}, illustrates the effect of bridge nodes on the accuracy of the communities found by CE-OR. The left-hand of panels show two layers of the network. The top right panel shows the communities obtained by the standard single network approach (C-SG-OR). The bottom right panel shows the communities obtained by our proposed CE-OR method.

Consider the community $C_{SG-OR}^2$ generated by C-SG-OR approach that has  edges (i, l), (h, m), (j, o) and (l, t) which are not intra-community edges in any of the  layers, and are present as bridge edges in only one of the layers. These edges will not be part of the metagraph and thus CE-OR does not know that they exist. CE-OR, thus, generates three communities $C_{CE-OR}^2$, $C_{CE-OR}^3$ and $C_{CE-OR}^4$, instead of merging them into one community, as per the C-SG-OR method.

However in the community $C_{SG-OR}^1$ generated by C-SG-OR, the edges (a, b) and (d, f) are bridge edges in one layer but are intra-community edges in another layer. Therefore these edges will be part of the metagraph. Thus CE-OR can use these edges and correctly generate the community $C_{CE-OR}^1$.

{\em The primary limitations} of our CE-AND and CE-OR  algorithms is due to the non-inclusion of bridge edges. In the AND-composed network, we rationalize this non-inclusion by positing that communities formed solely of bridge edges cannot be dense, and hence are not strong communities. In the OR-composed network, note that we only exclude an edge if it is a bridge edge in {\em all} the layers. This is an infrequent case where bridge nodes from all layers come together to form communities. Our experiments in Section \ref{sec:experiments}, justify this policy of not including bridge edges by demonstrating that the normalized mutual information (NMI) values between the communities returned by CE-AND and C-SG-AND are in general high. 

\section{Empirical Results}
\label{sec:experiments}

In this section we compare the performance and accuracy of our proposed algorithms with the ground truth results obtained by the standard methods, C-SG-AND and C-SG-OR.

\subsection{Experimental Setup:}


 Since the results of community detection depend heavily on the type of algorithm used~\cite{Abrahao:2012}, to control this parameter in the experiment we use the popular community detection algorithm {\em Infomap} \cite{InfoMap2014}, both to find the communities in the single network approach and the network decomposition approach.  Our algorithms were  implemented in C++ and were executed on a Linux machine with 8 GB RAM and installed with UBUNTU 16.10.

{\em Datasets Used.} We performed our experiments on multiplexes created from three real-world datasets and one synthetic dataset created using the RMAT~\cite{RMAT} graph generator. We selected the real-world datasets such that they were sufficiently large and contained communities. To test on larger networks with more vertices, we created the synthetic RMAT dataset. The details of the datasets are as follows (also see Table~\ref{tab:two});


\begin{itemize}
\item \textbf{\textit{IMDb:}} From the IMDB dataset~\cite{data/type2/IMDb}, we created the following three layers in the multiplex, with the nodes representing the actors. In the first layer, (L1, co-acting) two nodes  are connected if they co-acted in at least one movie. In the second layer,  (L2, rating) two nodes are connected if the average  ratings of the movies where they acted  were similar. In the third layer, (L3, genre) two nodes are connected if they acted in movies of similar genres. 

\item \textbf{\textit{DBLP:}} From the DBLP dataset of academic publications~\cite{data/type2/DBLP}, we selected all papers published from 2000-2018 in top three conferences VLDB (L1), SIGMOD(L2) and ICDM (L3). The nodes  were the authors. Two authors in each layer were connected if they had published a paper in the conference corresponding to the layer.

\item \textbf{\textit{Accident:}} From the dataset of road accidents that occurred in the United Kingdom in 2014 \cite{UKDataset2014}, we represented each  accident as a node.  Two nodes are connected in a layer if they occurred within 10 miles of each other and have similar Light (L1),  Weather (L2) or  Road Surface Conditions (L3).

\item \textbf{\textit{RMAT:}} The RMAT generator creates networks based on the Kronecker product of a matrix. We set the number of vertices to $2^{15}$ and the edges
to roughly eight times the number of vertices. We set the probabilities in each quadrant of the matrix as a=0.65, b=c=d=0.15 to create  a scale-free graph.

The first layer (L1) was the graph obtained by the generator. We applied cross perturbation to the other layers. That is we selected two edges (a, b) and (c, d), and replaced them with new edges (a, c) and (b, d).  Thus the number of edges remain the same, but the degree distribution and the structure of the graph changes.

In layer L2 we applied perturbation to 1\% of the edges and in layer L3 to 5\% of the edges. We limited the number of perturbations because if the network structure is significantly changed between the layers, the number of bridge edges also increase.

\end{itemize}

{\em Ground Truth and Accuracy Metrics:} Since our goal is to achieve the results obtained by the standard C-SG-AND and C-SG-OR methods, we use the communities obtained from these methods as the ground truth. We disregard communities of just one vertex, since these result due to an artifact of the algorithm rather than provide any meaningful analysis. We use two metrics to evaluate the accuracy of the communities - i) Normalized Mutual Information (NMI) that measures the quality with respect to the participating entity nodes only and ii)  modified-NMI that also takes into account the topology of the communities. For both metrics higher is better, with maximum value of 1 and minimum of 0(definitions in ~\cite{Labatut13}).

Each multiplex has 3 layers. Thus, a total of 4 compositions are possible (3 for 2-layers and 1 3-layers). Thus we compare results for 8 (4 combinations X 2 Boolean operations) composed networks.

\normalsize{
\begin{table}
\centering
\begin{tabular}{ |c||c|c|c|c| }
 \hline
 {Name} & {Vertices} & {Edges}  & {Edges}  & { Edges}  \\
      &          & { in L1} & { in L2}  & { in L3}   \\ \hline
 IMDB & 9,485 & 45,581 & 13,945,912 & 996,527 \\
 DBLP & 17,204 & 5,831 & 17,737 & 12,986 \\
 Accident & 5000 & 193,860 & 235,175 & 216,397 \\
 RMAT &  32,768 & 230,445 & 230,445 & 230,445 \\
 \hline

\end{tabular}

\caption {Summary of the sizes of the multiplexes.}
\label{tab:two}
\end{table}

}

\subsection{Accuracy of the Aggregation Algorithms.}

For the AND-composed networks we show in Figure \ref{fig:AND-acc}, the average NMI and m-NMI of all the four multiplexes with respect to the ground truth for the CV-AND and CE-AND methods. The results show that the \textit{accuracy obtained with CE-AND is higher than that from CV-AND.} 


\begin{table*}
\centering
\begin{tabular}{ |c||c|c||c|c||c|c||c|c| }
 \hline

  Multiplex & \multicolumn{2}{c||}{L1, L2} & \multicolumn{2}{c||}{L1, L3}  & \multicolumn{2}{c||}{L2, L3} & \multicolumn{2}{c|} {L1, L2, L3}\\
      &NMI & m-NMI &NMI & m-NMI &NMI & m-NMI &NMI & m-NMI \\ \hline
      \multicolumn{9}{|c|}{ Accuracy Values using CE-AND} \\ \hline
 IMDB & .97 & .93 &.98 & .97 & .88 &.86 &.99 &.99 \\
 DBLP & .92 & .84 &.99 & .96 & .98 &.96 & .98 & .95\\
 Accident &.96 & .98 & .94 & .98 & .91 &.96 & .88 & .95 \\
 RMAT &  .92 & .82 & .90 & .79 & .90 & .78 & .90 & .77 \\
 \hline
   \multicolumn{9}{|c|}{ Accuracy Values using CE-OR using Fractional Weights} \\ \hline
 IMDB & $<$.01 & $<$.01 & .97 & .99 & 1 & 1 & 1 & 1 \\
 DBLP & .83 & .79 & .87 & .80 & .75 & .60 & .71 & .56 \\
 Accident & .88 & .93 & .94 & .98 & .98 & .99 & .86 & .93 \\
 RMAT &  .74 & .64 & .76 & .59 & .75 & .55 & .73 & .54 \\
 \hline

\end{tabular}

\caption {Accuracy Values using CE-AND on the different compositions of the datasets.}
\label{tab:acc}
\end{table*}

\begin{figure}[h]
    \centering
  \includegraphics[width=\linewidth]{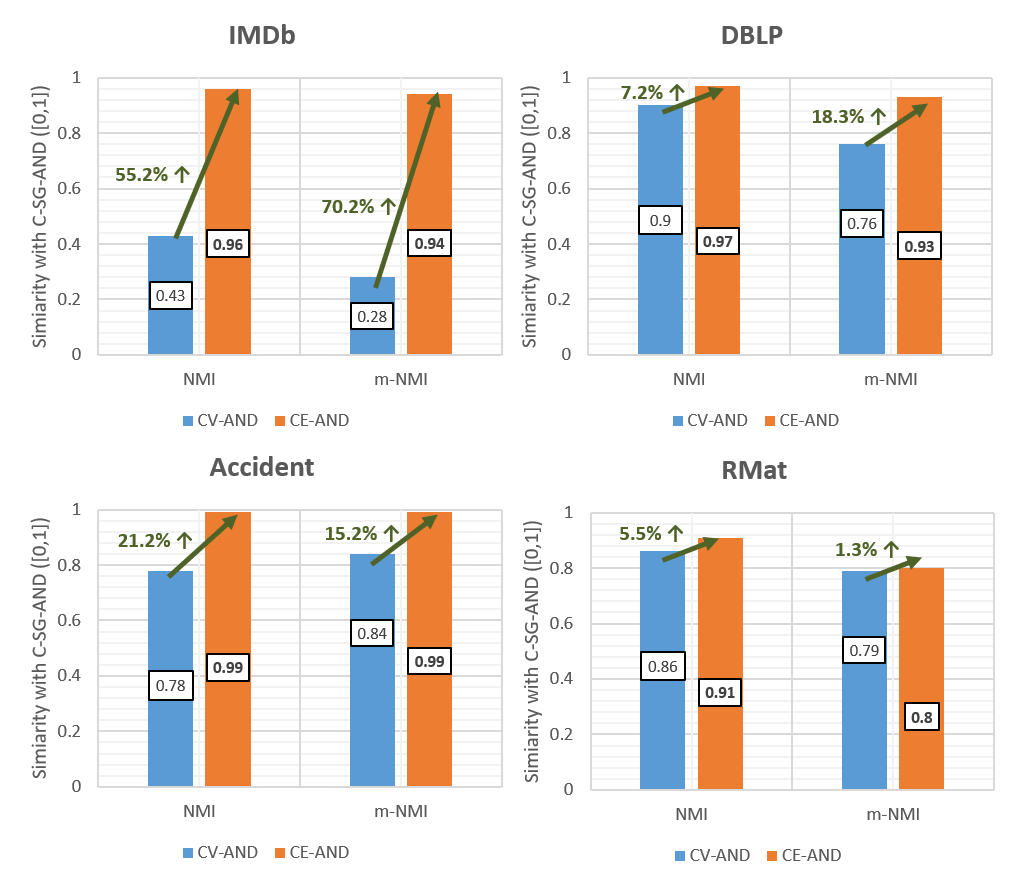}
    \caption{Comparison of Accuracy of CE-AND and CV-AND based on NMI and m-NMI.}
    \label{fig:AND-acc}
\end{figure}

For the OR-composed networks we show in Figure \ref{fig:OR-acc}, the average NMI and m-NMI of all the four multiplexes with respect to the ground truth for the two weighting metrics; Aggregation ($w_a$) and Fractional ($w_f$). The results show that the \textit{accuracy obtained using both the metrics are similar}.


\begin{figure}[h]
    \centering
  \includegraphics[width=\linewidth]{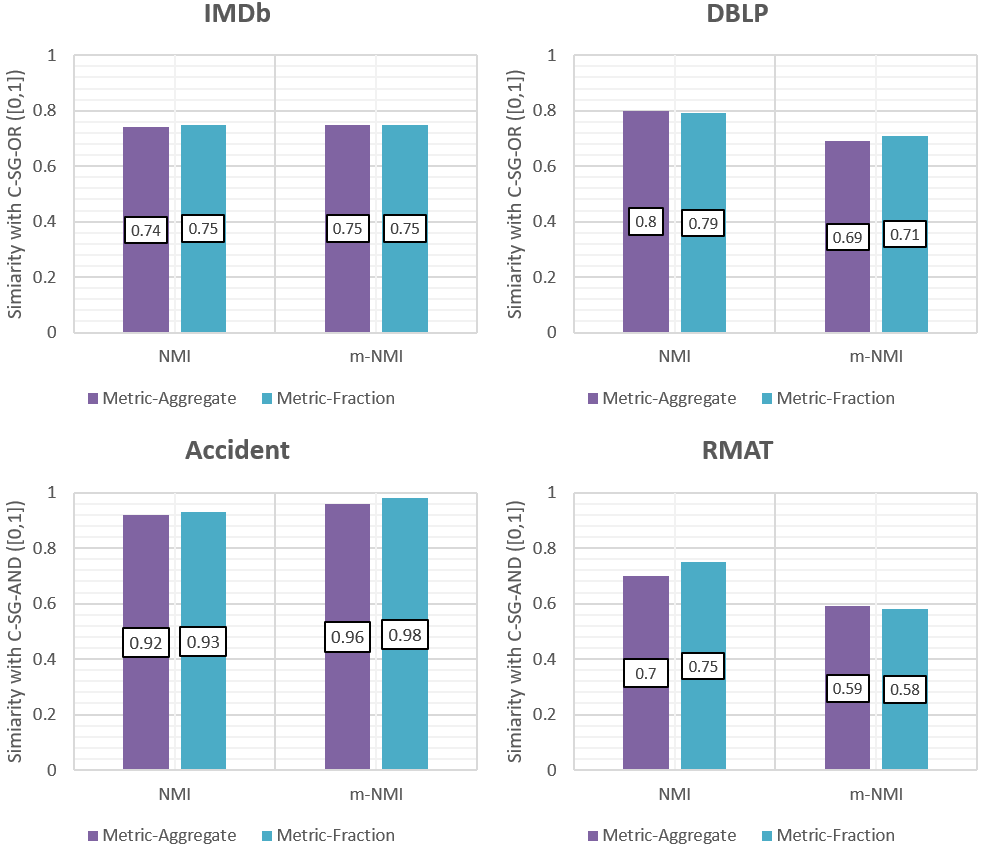}
    \caption{Accuracy of CE-OR with Different Weighting schemes  based on NMI and m-NMI.}
    \label{fig:OR-acc}
\end{figure}

In Table~\ref{tab:acc} we provide the accuracy values for all the different layer compositions with respect to CE-AND for the AND composition and CE-OR with Fractional Weights. As can be seen nearly all the values are high, $ \ge 70\%$.

Some low values occur for the CE-OR method. An egregious example is IMDb (L1, L2) for which the accuracy results are less than 1\%!  In this case the metagraph had 193 nodes, and on running the community detection algorithm 56 communities were obtained. However, the ground truth communities obtained by C-SG-OR had only 2 communities. This happened because there existed many bridge edges in the layers that were not included in the metagraph. Moreover, because the communities represented in the metanodes were small in size, the weights were also lower and could not combine the communities.


\subsection{Performance of the Aggregation Algorithms}
We now compare the time taken to obtain the communities using the aggregation methods (CV-AND, CE-AND and CE-OR) with respect to  C-SG-AND and C-SG-OR.

Figure \ref{fig:AND-time} shows that the time to compute the communities over all the 4-composed layers is significantly lower for both CV-AND and CE-AND methods than C-SG-AND. When the layers are sparse, CE-AND will be faster than CV-AND, as can be seen for DBLP multiplex. However if the network layers are dense, then the edge-based intersection approach of CE-AND has a higher cost as compared to the CV-AND.


\begin{figure}[h]
    \centering
  \includegraphics[width=\linewidth]{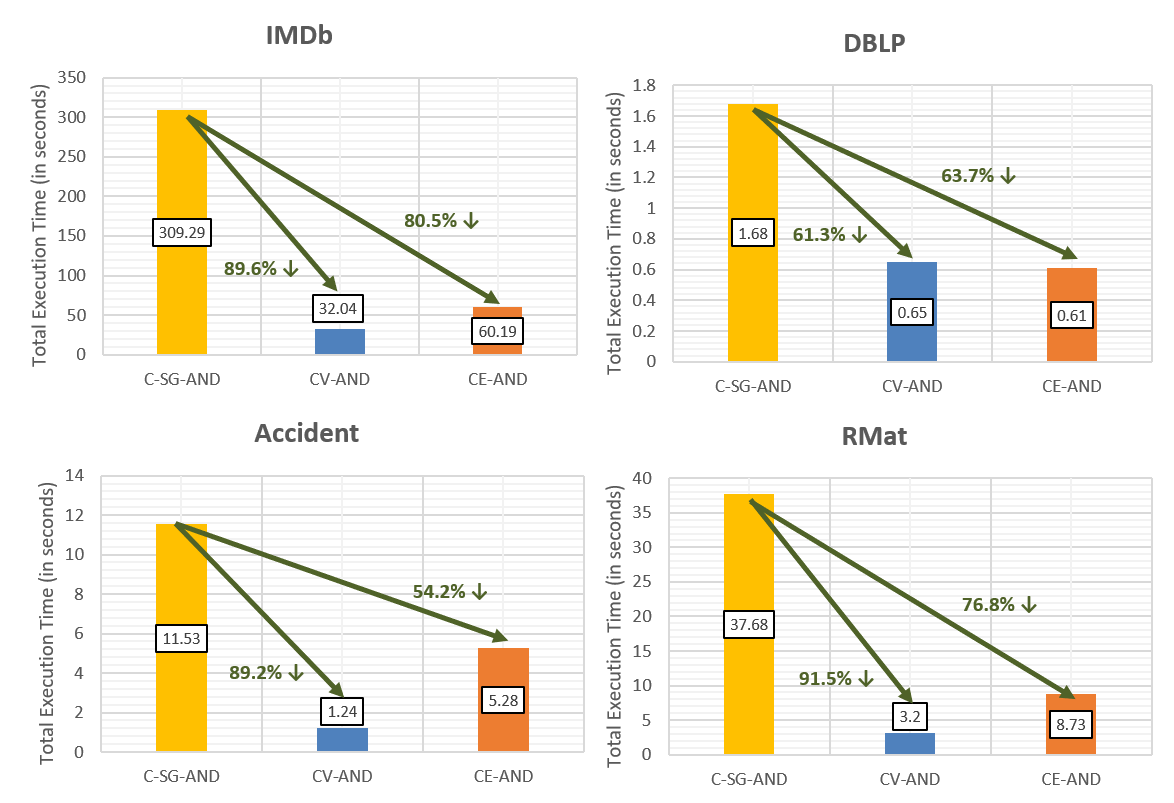}
    \caption{Efficiency of CV-AND and CE-AND as compared to C-SG-AND}
    \label{fig:AND-time}
\end{figure}

\begin{figure}[h]
    \centering
  \includegraphics[width=\linewidth]{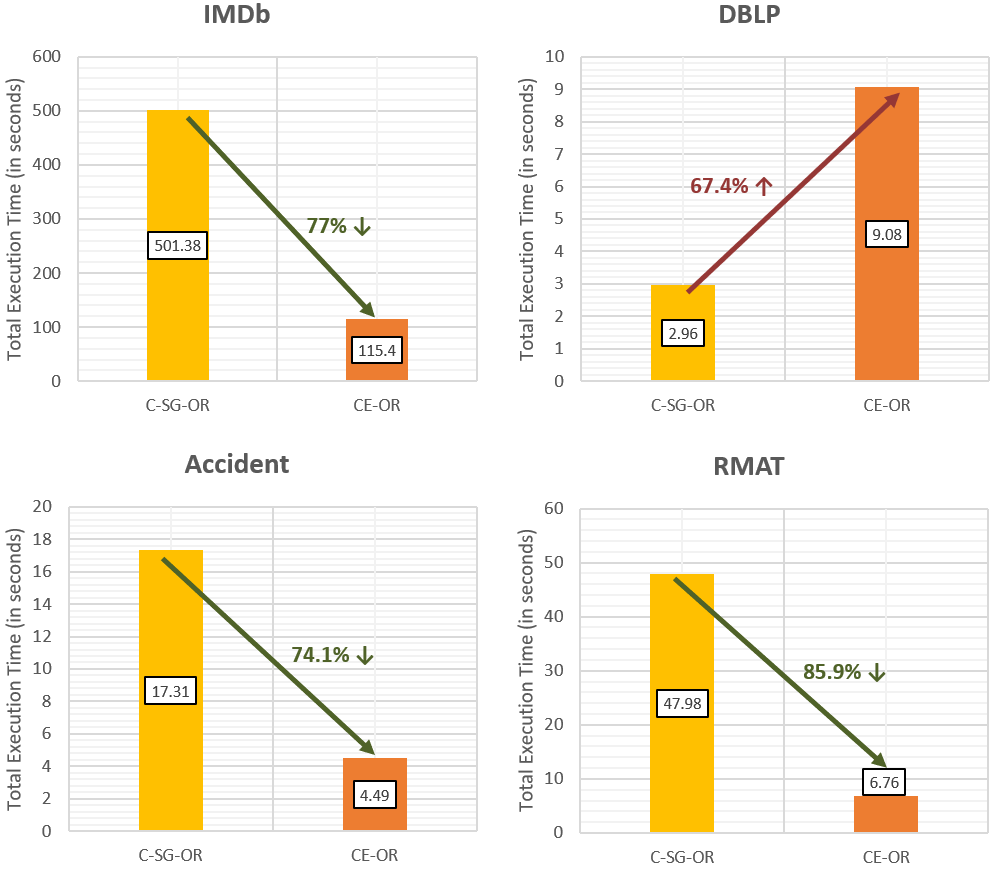}
    \caption{Efficiency of CE-OR as compared to C-SG-OR}
    \label{fig:OR-time}
\end{figure}

Figure~\ref{fig:OR-time} gives the time for executing CE-OR. For CE-OR, CE-AND is used as a subroutine. One scan of \textit{community edges} is required to generate the meta graph (OR-MG) on which we apply Infomap. If the layers are sparse and the multiplex contains many bridge nodes, then cost of generating the meta graph and applying Infomap will become an overhead as compared to simply applying Infomap on OR graph (C-SG-OR approach). This can be seen from the DBLP multiplex where sparse layers (density of densest layer (SIGMOD) = 0.0001!) make the CE-OR 67\% less efficient as compared to C-SG-OR. However, \textit{for multiplexes with fewer bridge edges (IMDb, Accident), CE-OR is significantly faster.} 

\section{Related Work}
\label{sec:relatedWork}



{\em Analysis of Multilayer Networks.}  Homogeneous multilayer networks~\cite{MultiLayerSurveyKivelaABGMP13,Boccaletti20141} are used to handle interactions among the same set of entities such as co-authorship  \cite{MinSubMulLayer2012}, citation across different topics \cite{ng2011multirank}, 
and relationships across different social media platforms \cite{magnani2013formation}.  Multiplexes have also been analyzed in form of  adjacency tensors\cite{de2013mathematical,maruhashi2011multiaspectforensics}. 

Community detection in temporal homogeneous multilayer networks use spectral optimization  of the modularity function  \cite{zhang2017modularity,bazzi2016community}. Techniques based on information theory have been proposed for reducing the number of layers in multilayer protein-protein interactions ~\cite{LayerAggDomenicoNAL14}. Software for creating or analysing multiplex networks include Muxviz~\cite{de2014muxviz,softwareMN/muxviz}, MAMMULT~\cite{battiston2014structural, softwareMN/MAMMULT}, MultiTensor~\cite{de2017community} and Pymnet~\cite{softwareMN/Pymnet}. However, they support only a few simple analysis functions, not community detection. 

\textit{Community Detection} is a well-studied problem in monoplex (single network) analysis.  A large number of competitive approximation algorithms exist (see reviews in 
\cite{Fortunato2009, Xie2013}) including algorithms for 
~\cite{Yang10}, weighted 
~\cite{Liu2011}, dynamic networks
~\cite{porterchaos13}, as well as parallel algorithms 
~\cite{StaudtSM14,BhowmickSrinivasan13}.

Recent work has also extended community detection algorithms for multilayer networks based on matrix factorization \cite{dong2012clustering}, 
pattern mining \cite{silva2012mining}, 
cluster expansion philosophy \cite{li2008scalable} and Bayesian probabilistic models \cite{xu2012model}, but all of these use the single layer approach.  
To the best of our knowledge,\textit{ we are the first to infer the communities of the combined network from communities of individual layers}.
\section{Conclusion and Future Work}
\label{sec:conclusions}
In this paper, we presented algorithms for efficiently finding communities in Boolean composed layers of  multiplex networks. The results show that for most cases our algorithms are significantly faster than the standard methods and produce results of similar quality. The only cases that our algorithm fails is when the layers have significantly more bridge edges.

In future plan, we will investigate how to include some percentage of bridge edges without reducing the computation time. We also plan to explore adaptive techniques that can select between the network decomposition and standard methods as suitable. Finally, we also aim to develop methods to efficiently find communities when the NOT operator is used.


%

\ifCLASSOPTIONcaptionsoff
  \newpage
\fi



\bibliographystyle{IEEEtran}
\bibliography{bibliography/santraResearch}
%

%


\end{document}